 \documentclass[smallabstract,smallcaptions]{dccpaper}

\usepackage{epsfig}
\usepackage{citesort}
\usepackage{amsmath}
\usepackage{amssymb}
\usepackage{color}
\usepackage{url}
\usepackage{float}
\usepackage{amsthm}

\DeclareMathOperator{\var}{var}
\DeclareMathOperator{\std}{std}

\newlength{\figurewidth}
\newlength{\smallfigurewidth} 

\floatstyle{ruled}
\newfloat{algorithm}{tbp}{loa}
\providecommand{\algorithmname}{Algorithm}
\floatname{algorithm}{\protect\algorithmname}

\theoremstyle{plain}
\newtheorem{thm}{\protect\theoremname}
\newtheorem{prop}[thm]{\protect\propositionname}
\providecommand{\propositionname}{Proposition}
\providecommand{\theoremname}{Theorem}

\begin{document}

\title
{\large
\textbf{Recover Subjective Quality Scores from Noisy Measurements}
}


\author{%
Zhi Li$^{\ast}$ and Christos G. Bampis$^{\dag}$\\[0.5em]
{\small\begin{minipage}{\linewidth}\begin{center}
\begin{tabular}{ccc}
$^{\ast}$Netflix & \hspace*{0.5in} &$^{\dag}$Department of ECE\\
100 Winchester Circle  && University of Texas at Austin\\
Los Gatos, CA 95032, USA && Austin, TX 78712, USA\\
\url{zli@netflix.com} && \url{bampis@utexas.edu}
\end{tabular}
\end{center}\end{minipage}}
}

\maketitle
\thispagestyle{empty}

\begin{abstract}
Simple quality metrics such as PSNR are known to not correlate well with subjective
quality when tested across a wide spectrum of video content or quality regime.
Recently, efforts have been made in designing objective quality metrics
trained on subjective data (e.g. VMAF), demonstrating better correlation with
video quality perceived by human. Clearly, the accuracy of
such a metric heavily depends on the quality of the subjective data
that it is trained on. In this paper, we propose a new approach to
recover subjective quality scores from noisy raw measurements, using maximum likelihood estimation, by
jointly estimating the subjective quality of impaired videos, the
bias and consistency of test subjects, and the ambiguity of video
contents all together. We also derive closed-from expression for the confidence interval of each estimate. 
Compared to previous methods which partially exploit
the subjective information, our approach is able to exploit
the information in full, yielding tighter confidence interval and better handling of outliers without 
the need for z-scoring or subject rejection. It also handles missing data 
more gracefully. Finally, as side information, it provides interesting insights
on the test subjects and video contents.
\end{abstract}

\section{Introduction}

In video coding research and development, two methods have generally
been used to evaluate the quality of impaired videos: subjective assessment
through viewer experiments and objective assessment using quality
metrics. Subjective assessment is the ultimate measure of viewer's
perception of quality, but is usually expensive to conduct. Very often,
objective assessment is used as an alternative or complement to report
perceptual quality. Peak-signal-to-noise-ratio (PSNR) and Structural
Similarity Index (SSIM)~\cite{ssim} are examples of objective quality
metrics originally designed for images but later extended to video.
Besides, an objective quality metric can also be used as an optimization
objective function, such as in per-title encode optimization~\cite{cbe-techblog}
and per-scene encode optimization under bitrate and video buffer constraints.

Simple metrics such as PSNR find success in evaluating small differences
and close performance among codecs or coding tools. However, it is
well known that they do not correlate well with quality perceived by human
when tested across a wide spectrum of quality or content.
Recently, efforts have been made in addressing this issue by designing
objective quality metrics through \emph{subjective data fusion} (VMAF, FVQA, VQM-VFD)~\cite{vmaf-techblog,joelin2,vqmvfd}. 
The basic idea is to extract low-level features or
elementary metrics that are quality-indicative, and then use a machine-learning
regressor, such as a support vector machine (SVM)~\cite{svm} or a neural net,
to fuse them into a ``meta-metric'' that makes a final prediction.
The regressor model is trained using subjective data, such as mean opinion scores (MOS) 
aggregated over the raw opinion scores
collected from subjective experiments. It is shown that the fusion-based
approach correlates better with subjective data than other approaches~\cite{vmaf-techblog}. 

Clearly, the accuracy of a fusion-based metric heavily depends
on the quality of the subjective data that it is trained on. Thus,
it is very important to provide clean and reliable training data.
On the other hand, raw opinion scores offered by viewers are often noisy and
unreliable, due to the following reasons:\vspace{-0.10in}
\begin{itemize}
\item \emph{Subject bias}. The notion of quality is highly subjective and
test subjects are entitled to rate the videos in their own opinions. For example, more picky
viewers tend to be biased toward lower scores, and vice versa.
Also, not every subject has ``golden eyes'' -- their sensitivity
to impairments varies. \vspace{-0.1in}
\item \emph{Subject inconsistency}. Subjective testing is a laborious process,
and not every viewer can maintain attentiveness throughout. Some tend
to rate more consistently than others.\vspace{-0.1in}
\item \emph{Content ambiguity}. Some contents tend to be more difficult to be rated than others. For example,
water surface with ripples in the dark is more ambiguous than a bright
blue sky.\vspace{-0.1in}
\item \emph{Outliers}. Last but not least, some raw scores are simply outliers
-- viewers may just not pay attention. Software issues may also render
scores meaningless.\footnote{We have seen real examples where some raw scores get misaligned with
subjects and contents due to a software bug!} \vspace{-0.10in}
\end{itemize}
Existing approaches address some of the issues above. For example,
MOS averages the raw scores from a number of subjects to produce an
aggregate score, compensating for the bias and inconsistency of individuals.
Z-score transformation (or z-scoring)~\cite{vqstudy10} normalizes the scores on a per-subject basis. 
Subject rejection~\cite{bt500} counts the number
of instances when a subject's rating looks like an outlier, and if
this occurs too often, the subject and all his or her scores are rejected. 

In this paper, we propose a new approach to recover subjective quality
scores from the noisy raw opinion measurements, by \emph{jointly estimating}
the subjective quality of impaired videos, the bias and consistency
of test subjects, and the ambiguity of video contents all together.
We propose a generative model that incorporates random variables representing
each of these factors, cast the problem as maximum likelihood estimation
(MLE), and derive a solution based on belief propagation (BP)~\cite{mackay}. We also derive 
closed-form expression for the confidence interval of each estimate based on Cramer-Rao bound~\cite{cover2006elements}. Compared to 
previous methods which partially exploit the subjective information, our approach is able to exploit the information in full, yielding better handling of outliers without 
the need for z-scoring or subject rejection. The resulting estimated subjective scores have a tighter 
confidence interval compared to conventional approaches. It also handles missing data 
more gracefully. Lastly, as side information, it provides interesting insights
on the test subjects and video contents.

The rest of the paper is organized as follows. Section \ref{sec:Related-Work}
discusses related work. Section \ref{sec:Problem-Formulation} defines
the problem and notations. Section \ref{sec:Traditional-Approaches}
describes traditional approaches, and Section \ref{sec:Proposed-Approach}
presents the proposed approach. Experimental
results are reported in Section \ref{sec:Results}. 

This work's open-source implementation can be found at~\cite{mleoss}.

\section{Related Work\label{sec:Related-Work}}

The ITU-R BT.500 Recommendation~\cite{bt500} defines procedures for
video subjective testings including both single-stimulus and double-stimulus
methods. It also defines the method for subjective rejection. The
ITU-T P.910 Recommendation~\cite{p910} defines procedures for calculating
differential MOS. Z-score transformation of subjective scores
is proposed in \cite{vqstudy10} as a pre-processing step prior to subject
rejection.

A theoretical model for subjects' influence on test scores is proposed in~\cite{janowski15}, which is similar in spirit to ours. The authors have focused on validating the model using real data, which also provides support for this work. However, they have relied on repetitive experiments for model parameter estimation, while our proposed solution based on BP is much more efficient.

Recovering subjective quality scores from noisy measurements is closely
related to the task of label inference from very large databases of
hand labeled images. To address the label inference problem, \cite{whitehill09} proposes
a solution based on probablistic graphical model. \cite{wang15} further
extends the idea to the task of image quality evaluation. There is
a number of differences between their approach and this work. First,
they formulate a classification problem whereas this work considers
a regression problem, which is closer to the nature of subjective
test procedure considered. Second, their work adopts a discriminative
model whereas this work adopts a generative model, allowing better
interpretability of results.

\section{Preliminaries\label{sec:Problem-Formulation}}

Consider an experiment with $S$ \emph{subjects}, indexed by $s=1,...,S$,
and $E$ impaired \emph{video encodes}, indexed by $e=1,...,E$. 
For simplicity, we consider the case of a single viewing session with no repititions. A
subject $s$ rates an impaired video encode $e$, producing a raw
opinion score $x_{e,s}$. In the \emph{full sampling} scenario, every subject rates every impaired
video. In the \emph{selective sampling} scenario, not every subject needs
to rate every impaired video -- if a score is missing, it is denoted
by $x_{e,s}=*$. In the rest of the paper, unless otherwise stated,
the full sampling scenario is assumed.

The raw opinion scores may be produced using
different test methods, including both single-stimulus and double-stimulus
methods: In \emph{absolute category rating (ACR)}, the subject is
instructed to watch the impaired video and give a rating on the scale
from 1 (quality is bad) to 5 (quality is excellent). In \emph{degradation
category rating (DCR)}, the subject is instructed to watch a pair
of videos -- the unimpaired reference video followed by the impaired
video, and then rate the impaired video on the scale from 1 (impairments
are very annoying) to 5 (impairments are imperceptible). In either method, 
unimpaired \emph{hidden reference} videos may be
present along with other impaired videos, and are also rated by the
subject. A differential score between the score of the impaired and its hidden 
reference may be used in place of the raw opinion score, and the resulting MOS calculated is called
differential MOS (or DMOS)~\cite{p910}. DMOS is
useful when the quality degradation from the reference is relevant,
rather than the absolute quality when, for example, the reference video is not perfect in quality.

Each impaired video
$e$ is also associated with a \emph{content} $c$, denoted by $c=\mathsf{c}(e)$,
with $c=1,...,C$ where $C$ is the total number of video contents
used in the experiment.

Throughout the paper, we use $\mu_{e}$ and $\mu_{s}$ to denote the
mean values calculated over scores for impaired video $e$ and for
subject $s$, respectively, i.e., $\mu_{e}=\frac{1}{S}\sum_{s}x_{e,s}$
and $\mu_{s}=\frac{1}{E}\sum_{e}x_{e,s}$. Similarly, $m_{n,e}$ and
$m_{n,s}$ denote the $n$-th order central moment over scores for
$e$ and $s$, respectively, i.e., $m_{n,e}=\frac{1}{S}\sum_{s}(x_{e,s}-\mu_{e})^{n}$
and $m_{n,s}=\frac{1}{E}\sum_{e}(x_{e,s}-\mu_{s})^{n}$. As special
case, the standard deviation over scores for $e$ and $s$ have the
form $\sigma_{e}=\sqrt{m_{2,e}}$ and $\sigma_{s}=\sqrt{m_{2,s}}$,
respectively.

\section{Traditional Approaches\label{sec:Traditional-Approaches}}

The most basic approach to recover subjective quality scores from
the raw opinion measurements is by averaging, or MOS: $x_{e}=\mu_{e}$,
for impaired video $e=1,...,E$. Before the averaging step, one can apply the following:

\emph{Z-score transformation}. It is advocated that the
scores undergo a normalization procedure on a per-subject basis~\cite{vqstudy10},
i.e., $x_{e,s}\leftarrow\frac{x_{e,s}-\mu_{s}}{\sigma_{s}}$. Applying
z-scoring could compensate for individual subject's bias and inconsistency,
preparing for favorable conditions for subject rejection. However, after z-scoring, the original scale
of the scores is unfavorably lost, leading to difficulty in interpreting
the results. Also, it only partially compensates for the influence of subjects. A theoretical analysis 
on z-scoring can be found in Appendix \ref{sec:analysis-z-score}.

\emph{Subject rejection}. In \cite{bt500}, a recommendation for
subject rejection is provided. The algorithm is reproduced in Algorithm
\ref{subjreject} for completeness. Video by video, the algorithm
counts the number of instances when a subject's opinion score deviates
by a few sigmas, and reject the subject if the occurences are more
than a fraction. All scores corresponding to the rejected subjects
are discarded, which could be an overkill.

\begin{algorithm}
\begin{itemize}
\item Input: $x_{e,s}$ for $s=1,...,S$ and $e=1,...,E$.\vspace{-0.15in}
\item Initialize $p(s)\leftarrow0$ and $q(s)\leftarrow0$ for $s=1,...,S$.\vspace{-0.15in}
\item For $e=1,...,E$:
\begin{itemize}\vspace{-0.15in}
\item Let $Kurtosis_{e}=\frac{m_{4,e}}{m_{2,e}^{2}}$. \vspace{-0.15in}
\item If $2\leq Kurtosis_{e}\leq4$, then $\epsilon_{e}=2$; otherwise $\epsilon_{e}=\sqrt{20}$.\vspace{-0.15in}
\item For $s=1,...,S$:
\begin{itemize}\vspace{-0.15in}
\item If $x_{e,s}\geq\mu_{e}+\epsilon_{e}\sigma_{e}$, then $p(s)\leftarrow p(s)+1.$\vspace{-0.1in}
\item If $x_{e,s}\leq\mu_{e}-\epsilon_{e}\sigma_{e}$, then $q(s)\leftarrow q(s)+1$.\vspace{-0.1in}
\end{itemize}\vspace{-0.1in}
\end{itemize}\vspace{-0.15in}
\item Initialize $Set_{rej}=\emptyset$.\vspace{-0.15in}
\item For $s=1,...,S$:\vspace{-0.15in}
\begin{itemize}\vspace{-0.05in}
\item If $\frac{p(s)+q(s)}{E}\geq0.05$ and $\left|\frac{p(s)-q(s)}{p(s)+q(s)}\right|<0.3$,
then $Set_{rej}\leftarrow Set_{rej}\cup\{s\}.$\vspace{-0.15in}
\end{itemize}\vspace{-0.15in}
\item Output: $Set_{rej}$.\vspace{-0.15in}
\end{itemize}
\caption{Subject rejection~\cite{bt500}}
\label{subjreject}

\end{algorithm}

\section{Proposed Approach\label{sec:Proposed-Approach}}

In this section, we describe the proposed approach of jointly estimating
the subjective quality of impaired videos, the bias and consistency
of subjects, and the ambiguity of video contents. Different from traditional 
approaches which require an explicit z-score transformation or a subject 
rejection step, the proposed model naturally accounts for subjective 
biases, inconsistencies and outliers.

\subsection{The Model}

We model the raw opinion scores as a random variable $\{X_{e,s}\}$
with the following form:
\begin{eqnarray}
X_{e,s} & = & x_{e}+B_{e,s}+A_{e,s},\label{eq:proposed}\\
B_{e,s} & \sim & \mathcal{N}(b_{s},v_{s}^{2}),\nonumber \\
A_{e,s} & \sim & \mathcal{N}(0,a_{c:\mathsf{c}(e)=c}^{2})\nonumber 
\end{eqnarray}
for $e=1,...,E$ and $s=1,...,S$. In this model, $x_{e}$ represents
the quality of impaired video $e$ perceived by an \emph{average }viewer,
$B_{e,s}$, $e=1,...,E$ are i.i.d. Gaussian variables representing
the factor of subject $s$, and $A_{e,s}$, $s=1,...,S$ and $e:\mathsf{c}(e)=c$
are i.i.d. Gaussian variables representing the factor of video content
$c$ (i.e., the content that $e$ corresponds to). The parameters
$b_{s}$ and $v_{s}^{2}$ represent the bias (i.e., mean) and inconsistency
(i.e., variance) of subject $s$, $s=1,...,S$. The parameter $a_{c}^{2}$
represents the ambiguity (i.e., variance) of content $c$, $c=1,...,C$.
In this formulation, the unknowns are the model parameters $\theta=(\{x_{e}\},\{b_{s}\},\{v_{s}\},\{a_{c}\})$, where $\{\cdot\}$ denotes the corresponding set.
The main idea of our approach is to jointly recover these unknowns
by MLE. Let $L=\log P(\{x_{e,s}\}|\theta)$ be the log likelihood
function, the goal is to solve for $\hat{\theta}=\arg\max_{\theta}L$.
Our ultimate goal is to recover the estimated scores $\{x_{e}\}$ of 
a hypothetical unbiased and consistent viewer, while $\{b_{s}\}$, $\{v_{s}\}$ 
and $\{a_{c}\}$ are side information about the subjects and video contents. An 
analysis of the observations and unknowns in (\ref{eq:proposed}) and its recoverability can be found in Appendix \ref{sec:recoverability}.

\subsection{Belief Propagation}

We derive a solution for the MLE formulation using
BP algorithm. We start with the log-likelihood function. By (\ref{eq:proposed}),
$X_{e,s}$ is a sum of a constant and independent Gaussian variables,
thus $X_{e,s}$ is also Gaussian with $X_{e,s}\sim \mathcal{N}(x_{e}+b_{s},v_{s}^{2}+a_{\mathsf{c}(e)}^{2})$.
The log-likelihood function can be expressed as:
\begin{eqnarray}
L(\theta) & = & \log P(\{x_{e,s}\}|\theta)\nonumber \\
 & = & \log P(\{x_{e,s}\}|\{x_{e}\},\{b_{s}\},\{v_{s}\},\{a_{c}\})\nonumber \\
 & = & \log\prod_{e,s}P(x_{e,s}|x_{e},b_{s},v_{s},a_{\mathsf{c}(e)})\label{eq:indep}\\
 & = & \sum_{e,s}\log P(x_{e,s}|x_{e},b_{s},v_{s},a_{\mathsf{c}(e)})\nonumber \\
 & \equiv & \sum_{e,s}-\frac{1}{2}\log\left(v_{s}^{2}+a_{\mathsf{c}(e)}^{2}\right)-\frac{1}{2}\cdot\frac{(x_{e,s}-x_{e}-b_{s})^{2}}{v_{s}^{2}+a_{\mathsf{c}(e)}^{2}}\label{eq:gaussian}
\end{eqnarray}
where (\ref{eq:indep}) uses the independence assumption on opinion
scores and (\ref{eq:gaussian}) uses the Gaussian formula with omission
of the constant terms. With (\ref{eq:gaussian}), the first- and second-order
partial derivatives of $L(\theta)$ with respect to parameters $x_{e}$,
$b_{s}$, $v_{s}$ and $a_{c}$ can be derived. We then apply the
Newton-Raphson rule~\cite{mackay} $\upsilon\leftarrow\upsilon-\frac{\partial L/\partial\upsilon}{\partial^{2}L/\partial\upsilon^{2}}$
to update each parameter $\upsilon$ at a time in each iteration.
Note that other update rules are also possible, but using the Newton-Raphson
rule can yield nice expressions with interpretability. Also note that the BP algorithm finds a local optimal solution when
the problem is nonconvex. It is important to initialize the parameters
properly. We choose the MOS as the initial values for $\{x_{e}\}$,
zeros for $\{b_{s}\}$, and the standard deviation values $\{\sigma_{s}\}$
and $\{\sigma_{c}\}$ for $\{v_{s}\}$ and $\{a_{c}\}$, respectively,
where $\sigma_{c}^{2}=\frac{\sum_{s,e:\mathsf{c}(e)=c}(x_{e,s}-\mu_{s})^{2}}{\sum_{s,e:\mathsf{c}(e)=c}1}$. The BP algorithm solution for the proposed MLE formulation is summarized
in Algorithm \ref{em}. The analytical forms of the update rules are
derived in Appendix \ref{sec:update-rules}. A good choice of refresh
rate and stop threshold are $\alpha=0.1$ and $\Delta x^{thr}=1e^{-9}$,
respectively.

\begin{algorithm}
\begin{itemize}
\item Input: \vspace{-0.15in}
\begin{itemize}\vspace{-0.05in}
\item $x_{e,s}$ for $s=1,...,S$ and $e=1,...,E$.\vspace{-0.15in}
\item Refresh rate $\alpha$.\vspace{-0.15in}
\item Stop threshold $\Delta x^{thr}$.\vspace{-0.15in}
\end{itemize}
\item Initialize $\{x_{e}\}\leftarrow\{\mu_{e}\}$, $\{b_{s}\}\leftarrow\{0\}$,
$\{v_{s}\}\leftarrow\{\sigma_{s}\}$, $\{a_{c}\}\leftarrow\{\sigma_{c}\}$. \vspace{-0.15in}
\item Loop:\vspace{-0.15in}
\begin{itemize}\vspace{-0.05in}
\item $\{x_{e}^{prev}\}\leftarrow\{x_{e}\}$.\vspace{-0.15in}
\item $b_{s}\leftarrow(1-\alpha)\cdot b_{s}+\alpha\cdot b_{s}^{new}$ where
$b_{s}^{new}=b_{s}-\frac{\partial L(\theta)/\partial b_{s}}{\partial^{2}L(\theta)/\partial b_{s}^{2}}$
for $s=1,...,S$.\vspace{-0.15in}
\item $v_{s}\leftarrow(1-\alpha)\cdot v_{s}+\alpha\cdot v_{s}^{new}$ where
$v_{s}^{new}=v_{s}-\frac{\partial L(\theta)/\partial v_{s}}{\partial^{2}L(\theta)/\partial v_{s}^{2}}$
for $s=1,...,S$.\vspace{-0.15in}
\item $a_{c}\leftarrow(1-\alpha)\cdot a_{c}+\alpha\cdot a_{c}^{new}$ where
$a_{c}^{new}=a_{c}-\frac{\partial L(\theta)/\partial a_{c}}{\partial^{2}L(\theta)/\partial a_{c}^{2}}$
for $c=1,...,C$.\vspace{-0.15in}
\item $x_{e}\leftarrow(1-\alpha)\cdot x_{e}+\alpha\cdot x_{e}^{new}$ where
$x_{e}^{new}=x_{e}-\frac{\partial L(\theta)/\partial x_{e}}{\partial^{2}L(\theta)/\partial x_{e}^{2}}$
for $e=1,...,E$.\vspace{-0.15in}
\item If $\left(\sum_{e=1}^{E}(x_{e}-x_{e}^{prev})^{2}\right)^{\frac{1}{2}}<\Delta x^{thr}$,
break.\vspace{-0.15in}
\end{itemize}
\item Output: $\{x_{e}\}$, $\{b_{s}\}$, $\{v_{s}\}$, $\{a_{c}\}$.\vspace{-0.15in}
\end{itemize}
\caption{BP solution for the proposed MLE formulation}
\label{em}
\end{algorithm}

\subsection{Confidence Interval}

The estimate of each model parameter $\{x_{e}\}$, $\{b_{s}\}$, $\{v_{s}\}$ and $\{a_{c}\}$ is associated with a confidence interval, which can be derived
using the Cramer-Rao bound~\cite{cover2006elements}. The asymptotic normal $95\%$ confidence interval for a parameter $\theta$ has the form:
\begin{eqnarray}
\hat{\theta}\pm1.96\frac{1}{\sqrt{-\frac{\partial^{2}L(\hat{\theta})}{\partial\theta^{2}}}},\label{eq:ci}
\end{eqnarray}
where $\hat{\theta}$ is the MLE of $\theta$. The closed-form expression for the second-order partial derivatives of $L(\theta)$ is derived in (\ref{eq:xe2nd}) $\sim$ (\ref{eq:ac2nd}) of Appendix \ref{sec:update-rules}. The derivation of (\ref{eq:ci}) can be found in Appendix ~\ref{sec:confidence-interval}.

\subsection{Generalization to Selective Sampling\label{sec:Generalization-to-Selective}}

All the algorithms described so far, including the traditional approaches
and the proposed MLE formulation, assumed the full sampling scenario. It is
not difficult to generalize the algorithms to selective sampling. Simply exclude the missing
terms during summation, i.e., use $\sum_{e:x_{e,s}\neq*}f(x_{e,s})$ instead of $\sum_{e}f(x_{e,s})$, for some function $f(\cdot)$. The mean and central
moment terms now become: $\mu_{e}=\frac{\sum_{s:x_{e,s}\neq*}x_{e,s}}{\sum_{s:x_{e,s}\neq*}1}$,
$\mu_{s}=\frac{\sum_{e:x_{e,s}\neq*}x_{e,s}}{\sum_{e:x_{e,s}\neq*}1}$,
$m_{n,e}=\frac{\sum_{s:x_{e,s}\neq*}(x_{e,s}-\mu_{e})^{n}}{\sum_{s:x_{e,s}\neq*}1}$
and $m_{n,s}=\frac{\sum_{e:x_{e,s}\neq*}(x_{e,s}-\mu_{s})^{n}}{\sum_{e:x_{e,s}\neq*}1}$.

\section{Results\label{sec:Results}}

We set up a number of experiments to evaluate the performance of the
proposed MLE method and compare it with a number of traditional methods,
including the plain MOS, MOS with subject rejection (SR-MOS) and MOS
with z-scoring and subject rejection (ZS-SR-MOS). 


We use raw opinion scores from two datasets: the Netflix Public (NFLX)
dataset~\cite{nflx-public} and the VQEG HD3 (VQEG) dataset~\cite{vqeg}. Refer
to Figure \ref{fig:Raw-opinion-scores} for a visualization of the
raw scores. The NFLX dataset includes four subjects whose raw scores
were scrambled due to a software issue during data collection. The
VQEG dataset includes various contents (SRC01-09 excluding 04 which
overlaps with the NFLX dataset) and streaming-relevant impairments
(HRC04, 07 and 16-21). Note that the SRC06-HRC07 video received 
very low scores due to encoding issues.

\begin{figure}
\begin{minipage}[t]{0.499\columnwidth}%
\begin{center}
\includegraphics[width=1\columnwidth]{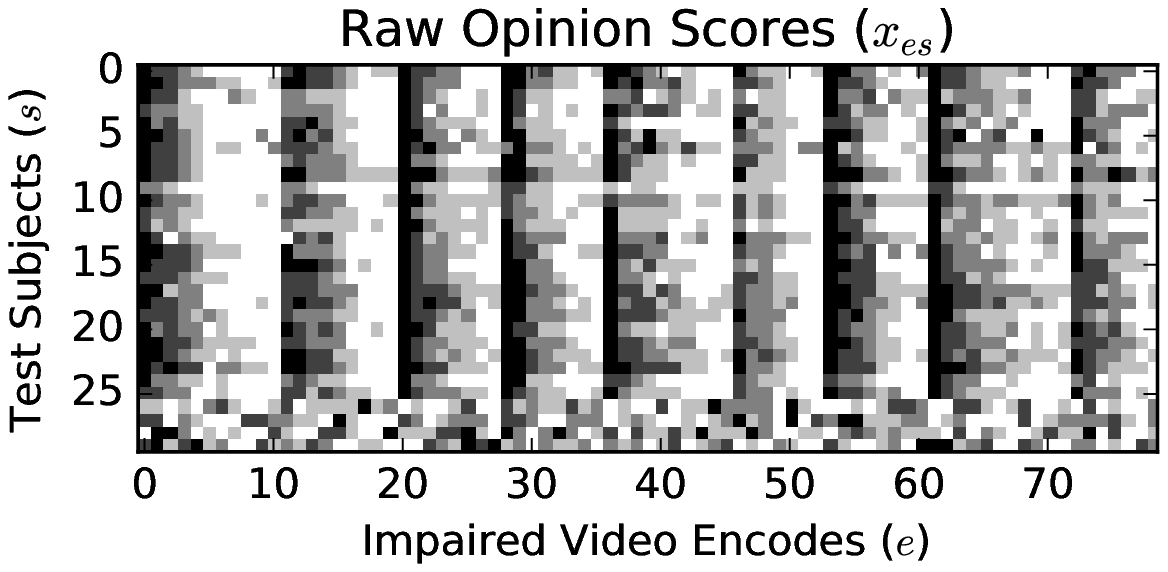}
\par\end{center}%
\end{minipage}%
\begin{minipage}[t]{0.499\columnwidth}%
\begin{center}
\includegraphics[width=1\columnwidth]{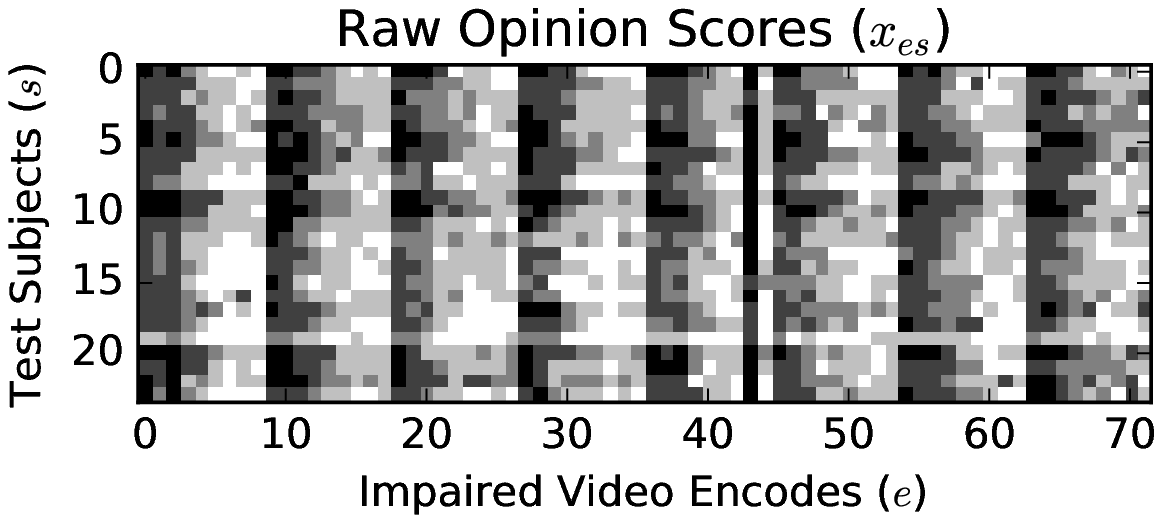}
\par\end{center}%
\end{minipage}

\vspace{-0.20in}

\caption{\label{fig:Raw-opinion-scores}Raw opinion scores from (left) the
NFLX dataset and (right) the VQEG dataset. Each pixel represents a
raw opinion score. The darker the color, the lower the score. The
impaired videos are arranged by contents, and within each content,
from low quality to high quality (with the reference video always
in the last). For the NFLX dataset, the last four rows correspond 
to corrupted subjective data.} 
\end{figure}

For many of the experiments, we do not have ``ground truth'' quality scores to compare against (we only
have noisy raw scores). Instead, we use the following methodology
in our report of results. For each recovery method, we have a benchmark
result, which is the recovered quality scores obtained \emph{using
that method} (for fairness) on an \emph{unaltered full dataset}. The
quality scores recovered under certain conditions (e.g., using a portion
of the raw scores, partially corrupted) is compared against the benchmark,
and a root-mean-squared-error (RMSE) value is reported. In doing so,
we could evaluate, for example, how fast the results converge toward
that benchmark result. In other experiments, we do have a ``ground
truth'', when artificially omitting a subject or creating a corruption on the data.

\subsection{An Example}

Let us first visually inspect one example result recovered by the
MLE and compare it with the MOS. Figure \ref{fig:Sampling-recovery-result}
shows the recovered parameters and their corresponding 95\% confidence intervals on the full NFLX dataset. Comparing
the result with Figure \ref{fig:Raw-opinion-scores} (left), a number of observations can be made: 
\begin{itemize}
\item The quality scores recovered by MLE are numerically different from the MOS, suggesting that the recovery is non-trivial. 
\vspace{-0.12in}\item The confidence intervals for the quality scores recovered by MLE is generally tighter, compared to the ones by MOS, suggesting an estimation with higher confidence. 
\vspace{-0.12in}\item Subject \#10 has the highest bias, which is evidenced by the whitish horizontal strip visible in Figure \ref{fig:Raw-opinion-scores} (left). 
\vspace{-0.12in}\item The last four subjects, whose raw scores were scrambled, have a very high $a_c$ value; correspondingly, their estimated bias have very loose confidence interval. \vspace{-0.15in}
\vspace{-0.15in}\item The content with the highest ambiguity is ElFuente2, which is the ``fountain and toddler'' scene, known to be difficult to evaluate. 
\end{itemize}
These observations demonstrate the potential of the MLE method on the problem at hand.

\begin{figure}
\begin{minipage}[t]{1\columnwidth}%
\begin{center}
\includegraphics[width=1\columnwidth]{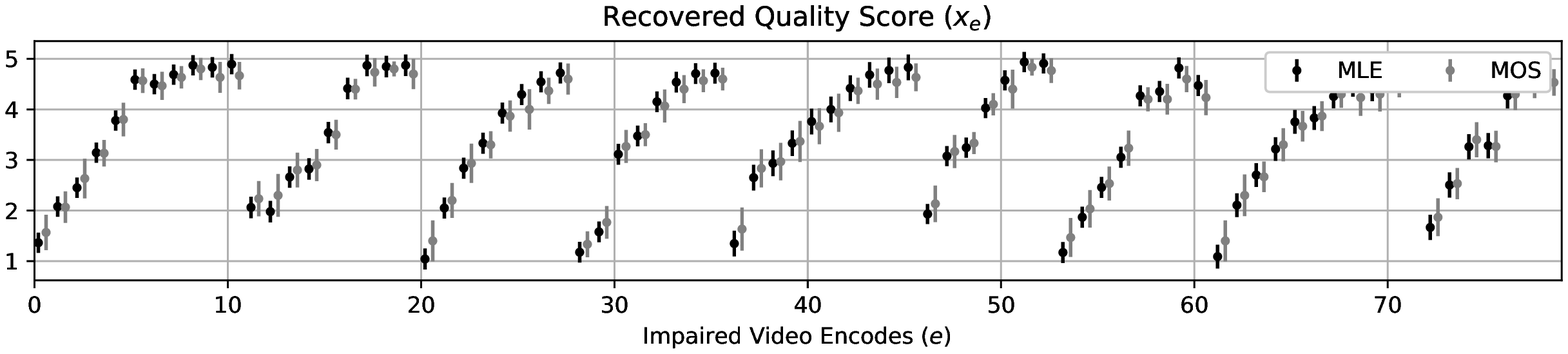}
\par\end{center}%
\end{minipage}

\begin{centering}
\vspace{-0.12in}
\par\end{centering}

\begin{minipage}[t]{0.499\columnwidth}%
\begin{center}
\includegraphics[width=1\columnwidth]{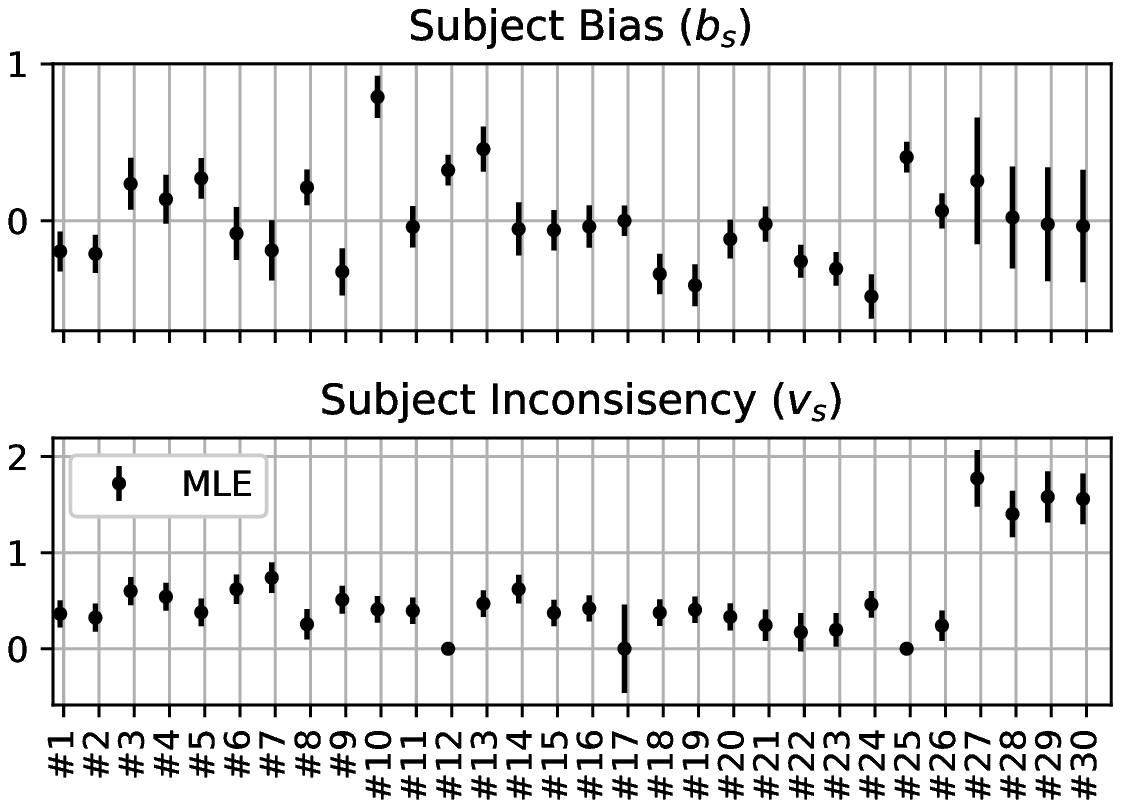}
\par\end{center}%
\end{minipage}%
\begin{minipage}[t]{0.499\columnwidth}%
\begin{center}
\includegraphics[width=1\columnwidth]{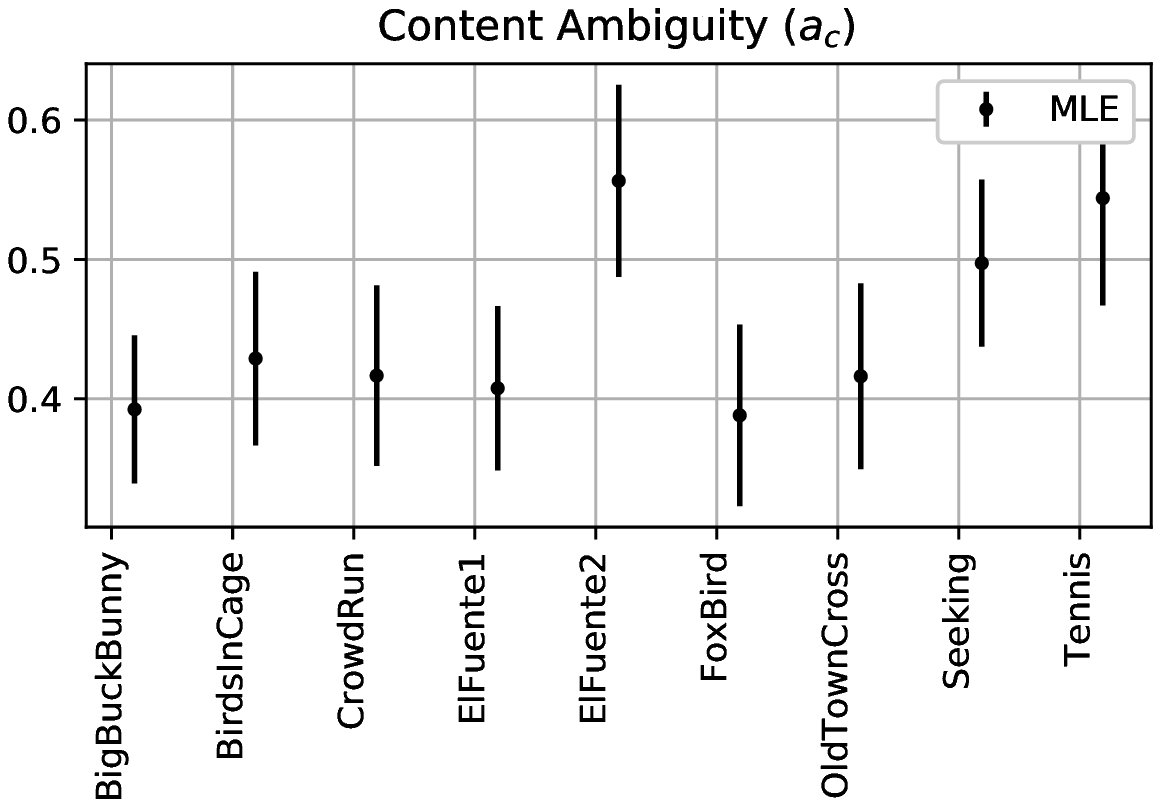}
\par\end{center}%
\end{minipage}

\vspace{-0.15in}

\caption{\label{fig:Sampling-recovery-result}Example recovery result on the
NFLX dataset. Both the estimated parameters and their 95\% confidence interval are shown.}
\end{figure}

\subsection{Convergence}

Next, we evaluate how fast MLE and other methods converge toward the
result recovered on the full dataset, as we increase the number of
test subjects. For each number of subjects, we randomly sample
among all the subjects, and repeated the experiment 100 times.
Figure \ref{fig:Convergence} illustrates the averaged RMSE as a function
of the subject numbers. It is shown that on the NFLX dataset, MLE has the closest quality scores to the full-dataset recovery than other methods for a given number of subjects. On the VQEG dataset, MLE has performance comparable to MOS and SR-MOS.

\begin{figure}

\begin{minipage}[t]{0.499\columnwidth}%
\begin{center}
\includegraphics[width=1\columnwidth]{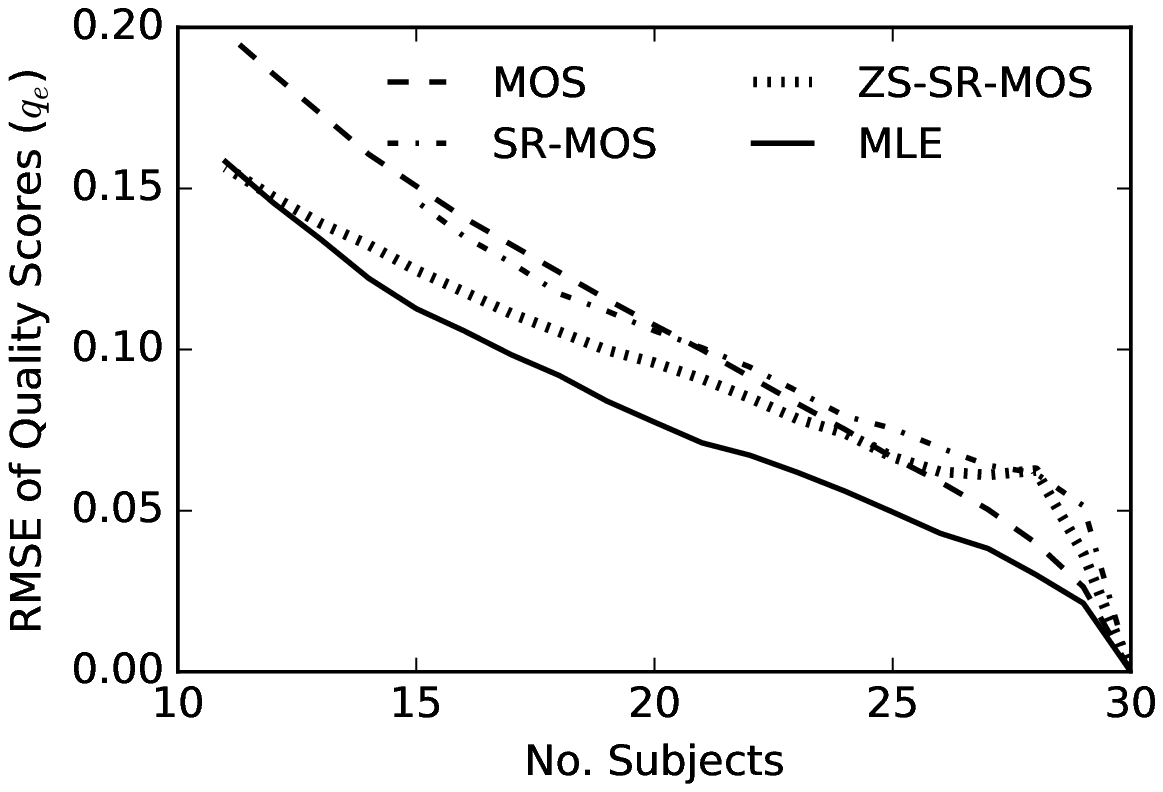}
\par\end{center}%
\end{minipage}%
\begin{minipage}[t]{0.499\columnwidth}%
\begin{center}
\includegraphics[width=1\columnwidth]{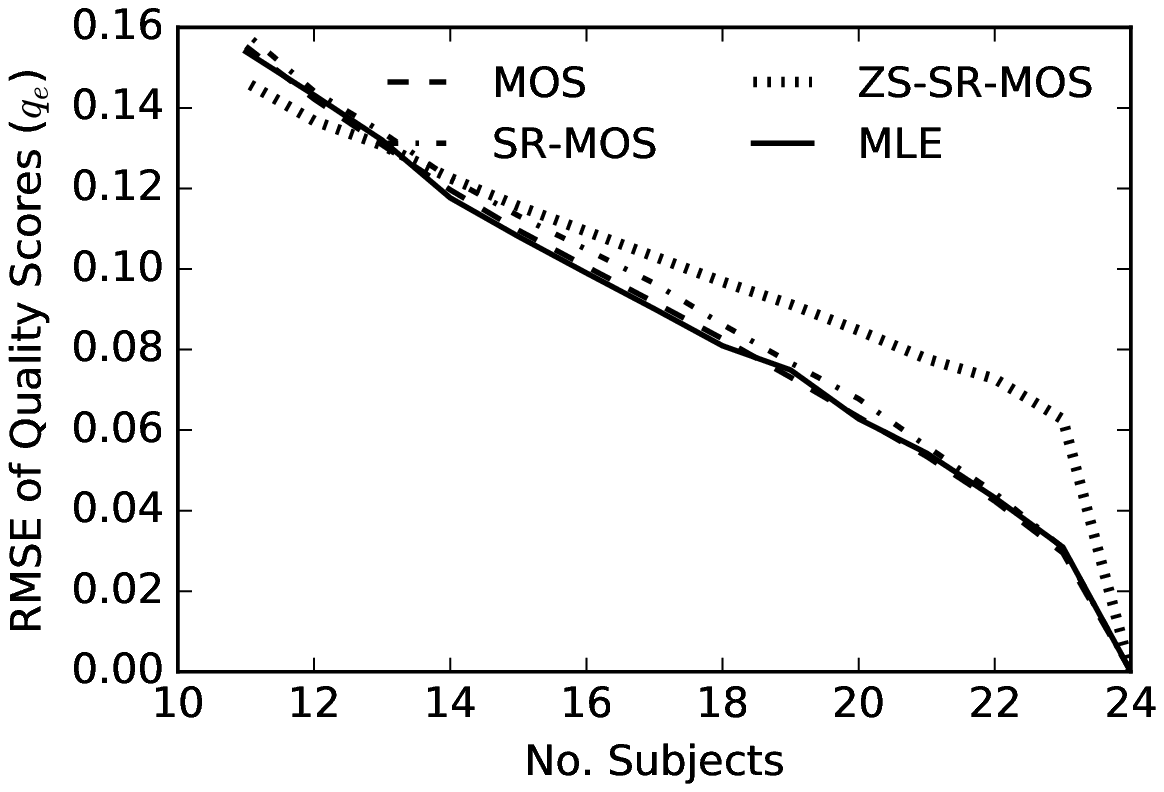}
\par\end{center}%
\end{minipage}

\vspace{-0.20in}

\caption{\label{fig:Convergence}RMSE of the recovered subjective quality scores
as a function of the number of test subjects, on (left) the NFLX dataset
and (right) the VQEG dataset.}\vspace{-0.15in}
\end{figure}

\subsection{Resistance to Corruption}

From the last section, it is evident that MLE shows faster convergence 
than other methods in the presence of data scrambling. To further 
corroborate our speculation, we evaluate how MLE and other
methods behave in the presence of data corruption. For each
dataset, we simulate two cases of corruption: a) subject corruption,
where all the scores corresponding to a number of subjects are scrambled,
and b) random corruption, where a raw score gets replaced by a random
score from 1 to 5 with a probability. Results for
a) and b) are reported in Figure \ref{fig:Resistence-to-Corruption}
and \ref{fig:Random-Corruption}, respectively. 

It can be observed that in the presence of subject corruption, MLE
achieves a substantial gain over the other methods, including the
ones with subject rejection. The reason is that the proposed 
model was able to capture the variance
of subjects explicitly and is able to compensate for it.
On the other hand, the traditional subject rejection scheme 
(Algorithm \ref{subjreject}) was only able to identify part of 
the corrupted subjects. It may also occur that only a subset 
of a subject's scores is unreliable. In that case, discarding all 
of the subject's scores is a waste of valuable subjective data. 
Meanwhile, traditional subject rejection employs a set of 
heuristic steps to determine outliers, which may lack 
interpretability. By contrast, the proposed model naturally 
integrates the various subjective effects together and is 
solved efficiently by our MLE method.

In the presence of random corruption, it can be seen that MLE does not show
any advantage over the other methods. This is because the proposed
model (\ref{eq:proposed}) is incapable of capturing this type of
corruption, hence it could not deal with it effectively.

\begin{figure}

\begin{minipage}[t]{0.499\columnwidth}%
\begin{center}
\includegraphics[width=1\columnwidth]{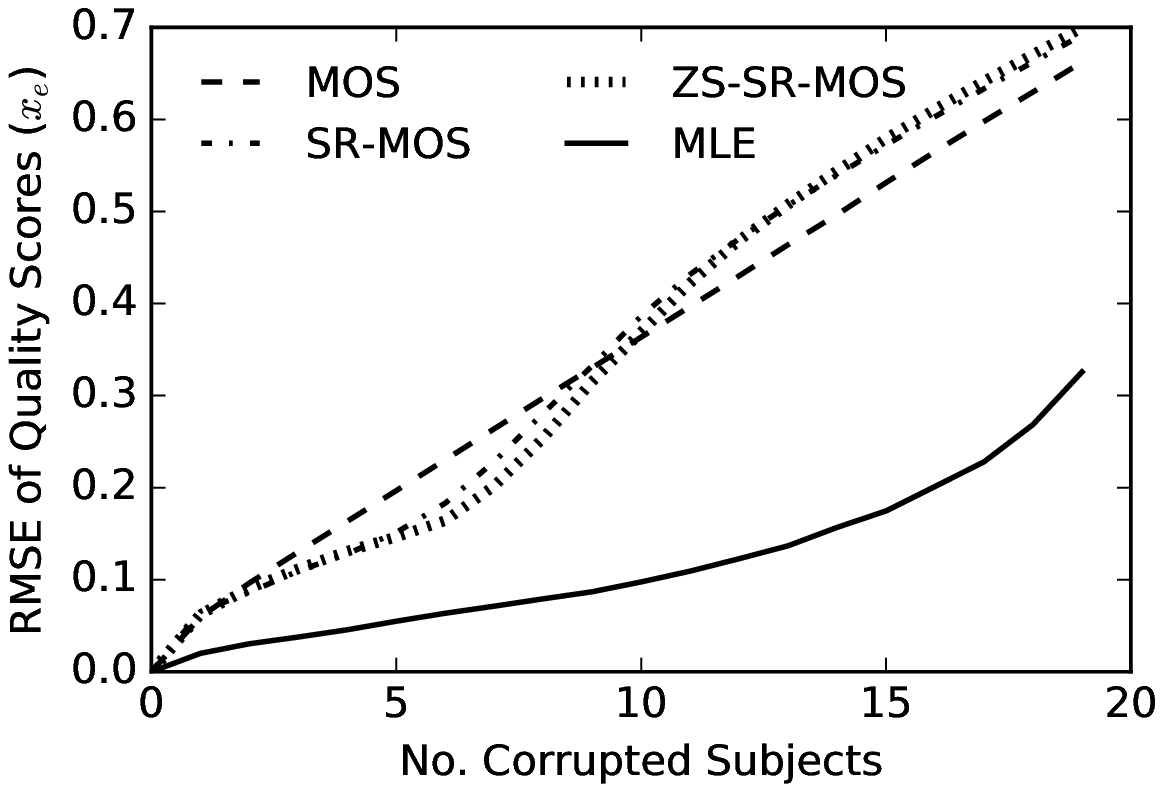}
\par\end{center}%
\end{minipage}%
\begin{minipage}[t]{0.499\columnwidth}%
\begin{center}
\includegraphics[width=1\columnwidth]{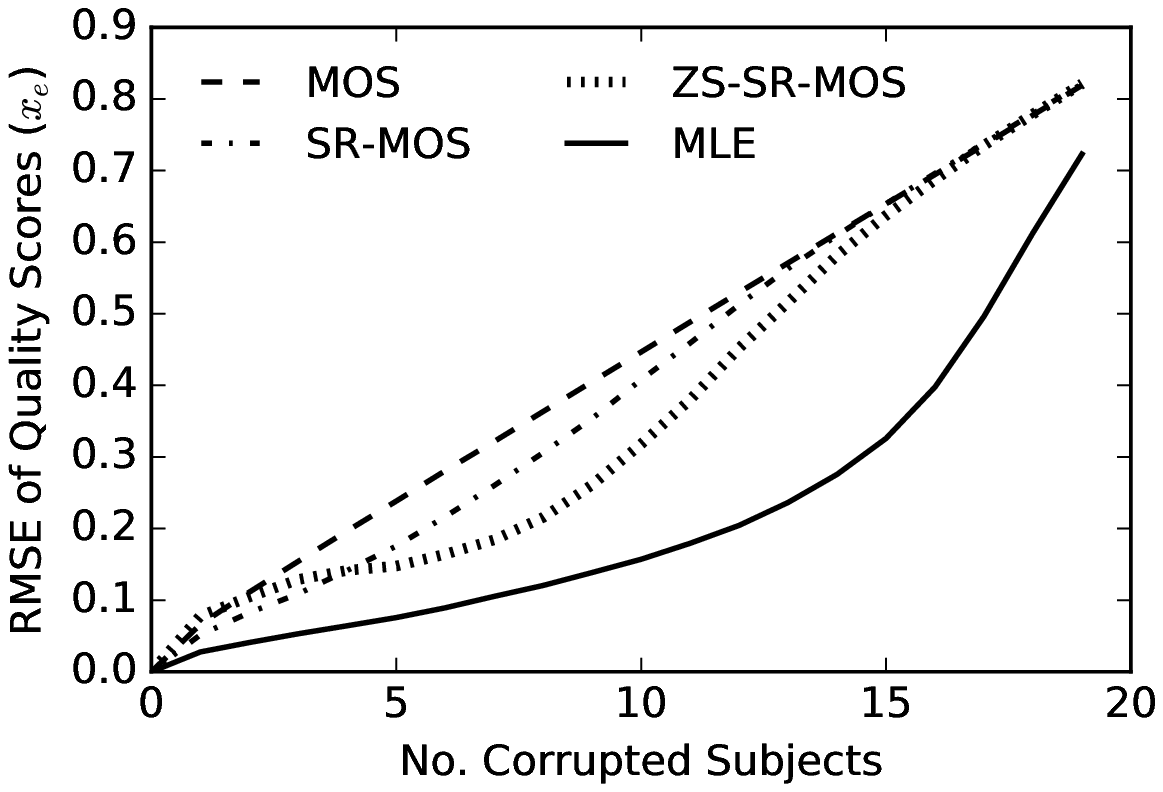}
\par\end{center}%
\end{minipage}

\vspace{-0.20in}

\caption{\label{fig:Resistence-to-Corruption}RMSE of the recovered subjective
quality scores as a function of the number of corrupted test subjects,
on (left) the NFLX dataset and (right) the VQEG dataset.}
\end{figure}

\begin{figure}

\begin{minipage}[t]{0.499\columnwidth}%
\begin{center}
\includegraphics[width=1\columnwidth]{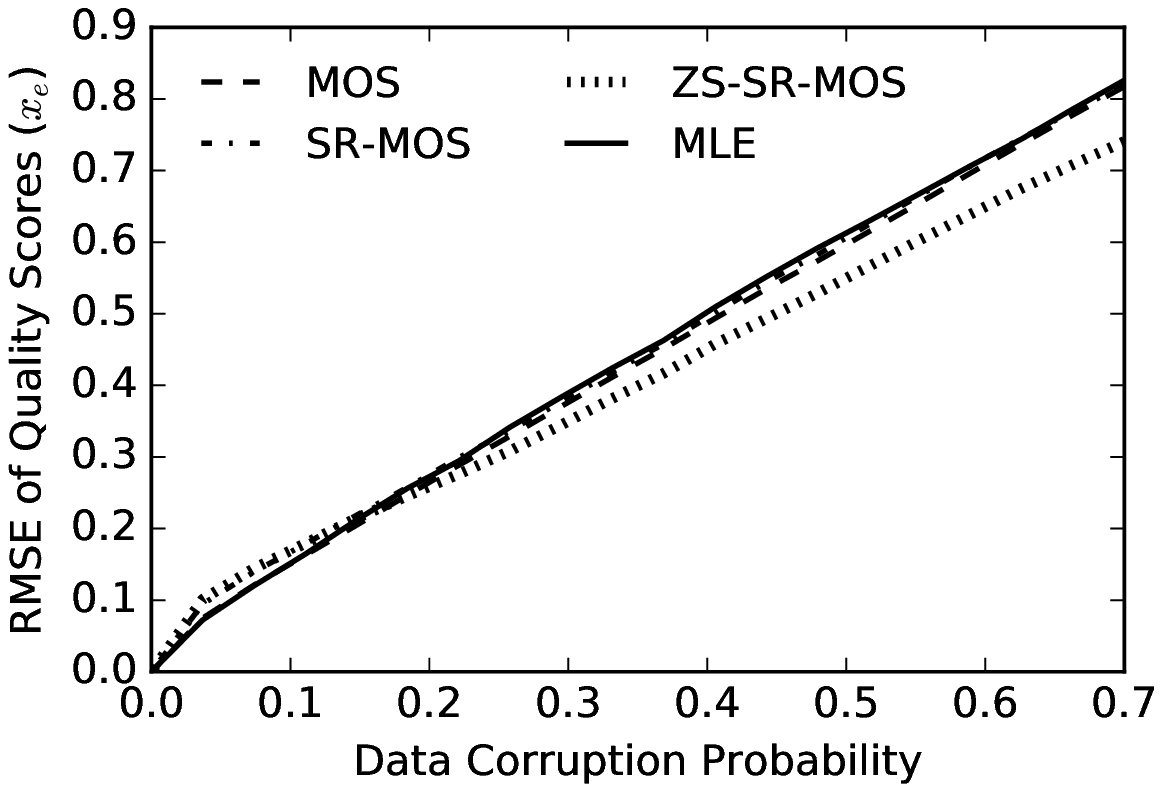}
\par\end{center}%
\end{minipage}%
\begin{minipage}[t]{0.499\columnwidth}%
\begin{center}
\includegraphics[width=1\columnwidth]{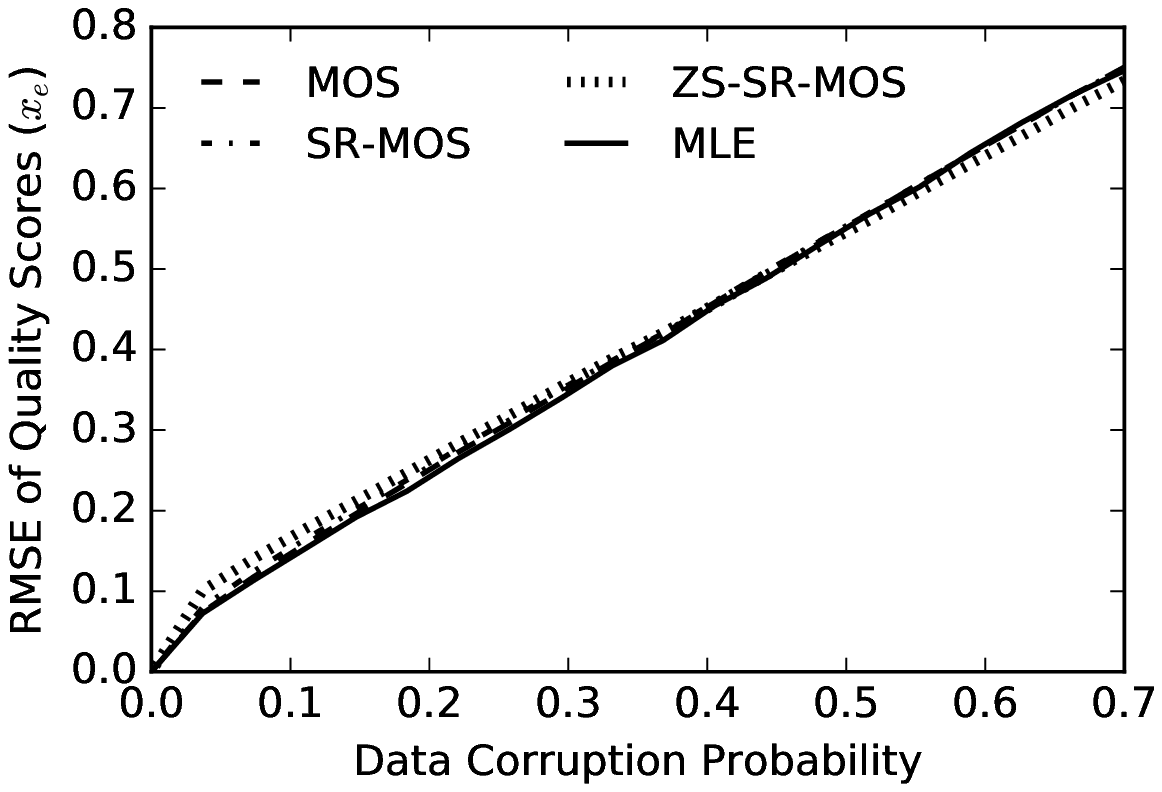}
\par\end{center}%
\end{minipage}

\vspace{-0.20in}

\caption{\label{fig:Random-Corruption}RMSE of the recovered subjective quality
scores as a function of the random score corruption probability, on
(left) the NFLX dataset and (right) the VQEG dataset. }
\end{figure}

\subsection{Selective Sampling}

We also evaluate the MLE and other methods under selective sampling,
where each subject only rates a part of all videos. We performed this step by assigning
each raw score a random probability of presence. As the probability
increases, more scores are sampled. The performance is reported in
Figure \ref{fig:Selective-Sampling}. On the NFLX dataset, again,
MLE has clear advantage over other methods. On the VQEG dataset, MLE
also shows gain over other methods. Since MLE accounts for the full
information, randomly missing some data points does not affect its predictive performance by much. By contrast, the MOS, SR-MOS, ZS-SR-MOS methods, which make local decisions on partial information, are greatly affected.

\begin{figure}

\begin{minipage}[t]{0.499\columnwidth}%
\begin{center}
\includegraphics[width=1\columnwidth]{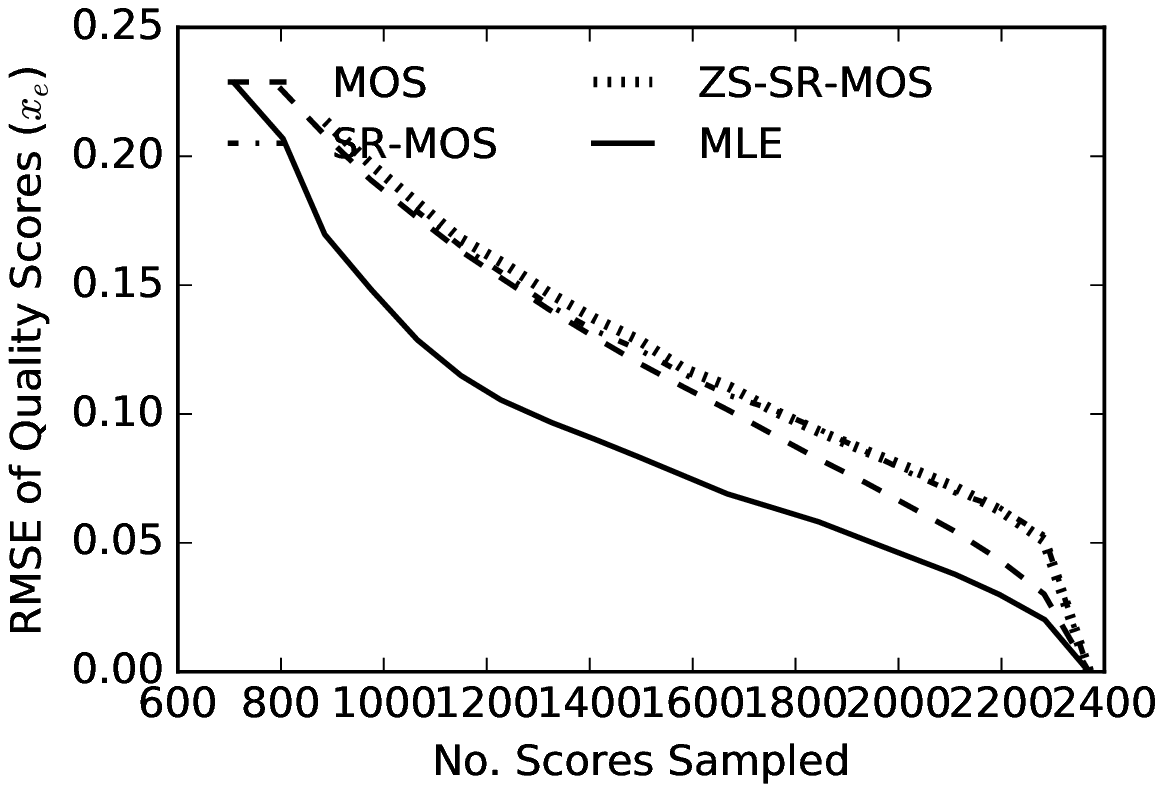}
\par\end{center}%
\end{minipage}%
\begin{minipage}[t]{0.499\columnwidth}%
\begin{center}
\includegraphics[width=1\columnwidth]{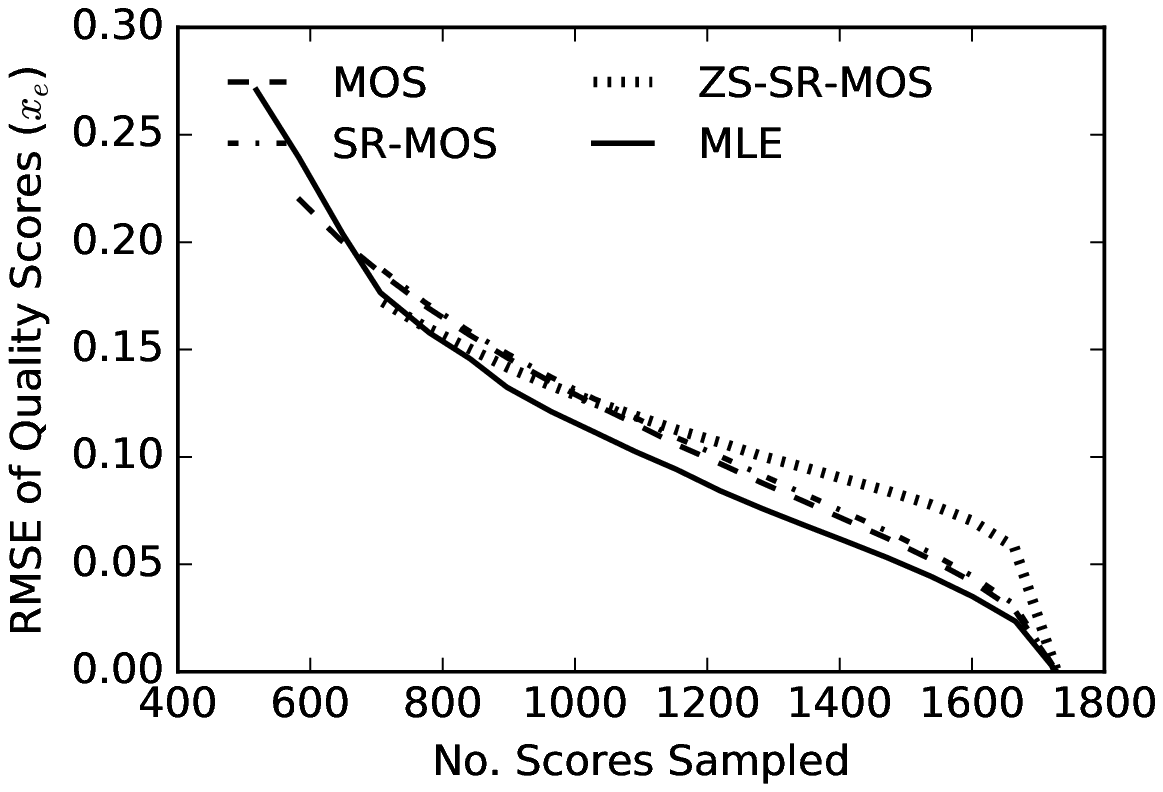}
\par\end{center}%
\end{minipage}

\vspace{-0.20in}

\caption{\label{fig:Selective-Sampling}RMSE of the recovered subjective quality
scores as a function of the number of scores randomly sampled, on
(left) the NFLX dataset and (right) the VQEG dataset. }
\end{figure}

\section{Summary and Future Work\label{sec:Conclusions}}

We have presented a new approach to process raw opinion scores collected
in video quality subjective testing, by using a generative model
that jointly captures the quality of impaired videos, bias and inconsistency
of subjects and the ambiguity of contents. We determined the model parameters
by formulating a MLE problem and devising a belief propagation solution. It was shown
that the recovered parameters were able to capture the subjective effects (bias, inconsistencies etc.) of the problem at hand, and that the proposed solution outperformed other methods in terms of its resistance to subject corruption, tighter confidence interval, better handling of missing data and provision of side information on the test 
subjects and video contents.

We list a number of directions to further this work: 1) The
recovered side information on the test subjects and video contents can
be used actively in the process to train an objective quality metric
to yield more accurate prediction. 2) The proposed model can
be further extended to other test methods, such as pairwise comparison.
3) A theoretical analysis on the belief propagation algorithm can uncover its region of convergence.

\Section{References}
\bibliographystyle{IEEEtran}
\bibliography{refs}

\appendix

\section{Analysis of Z-score Transformation\label{sec:analysis-z-score}}

Let us assume a simplified version of model (\ref{eq:proposed}):
\begin{eqnarray*}
X_{e,s} & = & x_{e}+B_{e,s}\\
B_{e,s} & \sim & \mathcal{N}(b_{s},v_{s}^{2}).
\end{eqnarray*}
Let $y_{e,s}=x_{e,s}-\mu_{s}$ be the mean-subtracted version of $x_{e,s}$,
and let $B_{e,s}=b_{s}+v_{s}N_{e,s}$ where $N_{e,s}\sim\mathcal{N}(0,1)$.
The corresponding random variable $Y_{e,s}$ has the form:

\begin{eqnarray}
Y_{e,s} & = & X_{e,s}-\frac{1}{E}\sum_{e'}X_{e',s}\nonumber \\
 & = & x_{e}+b_{s}+v_{s}N_{e,s}-\frac{1}{E}\sum_{e'}(x_{e'}+b_{s}+v_{s}N_{e's})\nonumber \\
 & = & (x_{e}-\frac{1}{E}\sum_{e'}x_{e'})+v_{s}N_{e,s}-\frac{1}{E}v_{s}\sum_{e'}N_{e's}\nonumber \\
 & = & y_{e}+\frac{E-1}{E}v_{s}N_{e,s}-\frac{1}{E}v_{s}\sum_{e':e'\neq e}N_{e's}\label{eq:ye}\\
 & \simeq & y_{e}+v_{s}N_{e,s},\label{eq:ye2} 
\end{eqnarray}
where in (\ref{eq:ye}) we define $y_{e}=x_{e}-\mu_{x}$ and $\mu_{x}=\frac{1}{E}\sum_{e'}x_{e'}$, and (\ref{eq:ye2}) is due to $E\gg1$.
Z-score transformation computes the z-score $z_{e,s}=(x_{e,s}-\mu_{s})\left/\sigma_{s}\right.=y_{e,s}\left/\sqrt{\frac{1}{E}\sum_{e'}y_{e',s}^{2}}\right.$.
The corresponding random variable $Z_{e,s}$ has the form:
\begin{eqnarray}
Z_{e,s} & = & \frac{Y_{e,s}}{\sqrt{\frac{1}{E}\sum_{e'}Y_{e',s}^{2}}}\nonumber \\
 & \simeq & \frac{y_{e}+v_{s}N_{e,s}}{\sqrt{\frac{1}{E}\sum_{e'}(y_{e'}+v_{s}N_{e',s})^{2}}}\nonumber \\
 & \simeq & \frac{y_{e}+v_{s}N_{e,s}}{\sqrt{\frac{1}{E}\sum_{e'}y_{e'}{}^{2}+\frac{1}{E}v_{s}^{2}\sum_{e'}N_{e',s}^{2}}}\label{eq:yep} \\
 & = & \frac{x_{e}-\mu_{x}+v_{s}N_{e,s}}{\sqrt{\sigma_{x}^{2}+\frac{1}{E}v_{s}^{2}\sum_{e'}N_{e',s}^{2}}}\label{eq:sigmae}\\
 & \neq & \frac{x_{e}-\mu_{x}}{\sigma_{x}},\nonumber 
\end{eqnarray}
where (\ref{eq:yep}) is due to that $\sum_{e}y_{e}=0$ and in (\ref{eq:sigmae}) we define $\sigma_{x}=\sqrt{\frac{1}{E}\sum_{e}y_{e}^{2}}$.
The last inequality suggests that there are non-vanishing terms related
to subjects that cannot be canceled out by z-score transformation.
In other words, z-scoring can only partially compensate subjects'
bias and inconsistency.

\section{Analysis of Observations and Unknowns in (\ref{eq:proposed})\label{sec:recoverability}}

We consider the observations and unknowns in  (\ref{eq:proposed}). Assume
a typical scenario, where the subjective experiment has $S=30$ subjects,
$C=20$ contents, and $E=200$ impaired videos. The number of unknowns
is thus $E+2S+C=280$ and the number of observations is $ES=6000$.
At a first glance, we may conclude that the number of observations
are much greater than the number of unknowns, and the problem has
a well-conditioned solution.

Sounds pretty good? Let's take a closer look. This time, consider a simplified model 
by removing the variation terms $\{v_{s}\}$ and $\{a_{c}\}$, i.e., 
\[
x_{e,s}=x_{e}+b_{s}.
\]
After the variation terms are removed, everything is deterministic.
Since the relationship is linear, we hope to find a least-squares
solution of a linear system with $ES$ observations and $E+S$ unknowns.
Let $u=[\{x_{e}\}_{e=1}^{E},\{b_{s}\}_{s=1}^{S}]{}^{T}$ be the vector
of unknowns, and $o=[\{x_{e,s}\}]{}^{T}$ be the vector of observations,
and let $\Omega=[\Omega_{i,j}]$ be a $ES\times(E+S)$ matrix, where $\Omega_{e+(s-1)E,e}=1$
and $\Omega_{e+(s-1)E,s+E}=1$ for $e=1,...,E$
and $s=1,...,S$, and the rest of the elements are all zeros. We look
to find a least-squares solution to the following linear system:
\[
\Omega\cdot u=o.
\]

\begin{prop}
\label{prop:omega}Matrix $\Omega$ has rank $E+S-1$.\end{prop}
\begin{proof}
Apply Gaussian elimination. The first $E$ rows are independent.
Subtract the $(e+(s-1)E)$-th row by the $e$-th row, $e=1,...,E$
and $s=1,..,S$, we can get another $(S-1)$ independent rows. In
total, there are $E+S-1$ independent rows in $\Omega$.
\end{proof}
Proposition \ref{prop:omega} implies that we are missing one equation, and without it there are infinite
number of solutions. Fortunately, it is not too difficult to come
up with one additional equation that is reasonable. For example, we
can let the first impaired video have a score equal to its MOS, i.e.,
\[
x_{e=1}=\mu_{e=1}.
\]
We can readily find the least-squares solution after adding the additional condition, 
which is also necessary for solving the full model (\ref{eq:proposed}).

\section{Derivation of Update Rules in Algorithm \ref{em}\label{sec:update-rules}}

This section derives the update rules used in Algorithm \ref{em}. 
Consider $L(\theta)$ derived in (\ref{eq:gaussian}) where $\theta=(\{x_{e}\},\{b_{s}\},\{v_{s}\},\{a_{c}\})$.
The first-order partial derivatives of $L(\theta)$ are expressed
as follows:

\begin{eqnarray*}
\frac{\partial L(\theta)}{\partial x_{e}} & = & \sum_{s}\frac{x_{e,s}-x_{e}-b_{s}}{v_{s}^{2}+a_{\mathsf{c}(e)}^{2}}\\
\frac{\partial L(\theta)}{\partial b_{s}} & = & \sum_{e}\frac{x_{e,s}-x_{e}-b_{s}}{v_{s}^{2}+a_{\mathsf{c}(e)}^{2}}\\
\frac{\partial L(\theta)}{\partial v_{s}} & = & \sum_{e}-\frac{v_{s}}{v_{s}^{2}+a_{\mathsf{c}(e)}^{2}}+\frac{v_{s}(x_{e,s}-x_{e}-b_{s})^{2}}{\left(v_{s}^{2}+a_{\mathsf{c}(e)}^{2}\right)^{2}}\\
\frac{\partial L(\theta)}{\partial a_{c}} & = & \sum_{s,e:\mathsf{c}(e)=c}-\frac{a_{\mathsf{c}(e)}}{v_{s}^{2}+a_{\mathsf{c}(e)}^{2}}+\frac{a_{\mathsf{c}(e)}(x_{e,s}-x_{e}-b_{s})^{2}}{\left(v_{s}^{2}+a_{\mathsf{c}(e)}^{2}\right)^{2}}
\end{eqnarray*}
Repeat the differentiation and we get the second-order partial derivatives:

\begin{eqnarray}
\frac{\partial^{2}L(\theta)}{\partial x_{e}^{2}} & = & \sum_{s}-\frac{1}{v_{s}^{2}+a_{\mathsf{c}(e)}^{2}}\label{eq:xe2nd}\\
\frac{\partial^{2}L(\theta)}{\partial b_{s}^{2}} & = & \sum_{e}-\frac{1}{v_{s}^{2}+a_{\mathsf{c}(e)}^{2}}\label{eq:bs2nd}\\
\frac{\partial^{2}L(\theta)}{\partial v_{s}^{2}} & = & \sum_{e}-\frac{a_{\mathsf{c}(e)}^{2}-v_{s}^{2}}{\left(v_{s}^{2}+a_{\mathsf{c}(e)}^{2}\right)^{2}}+\frac{(x_{e,s}-x_{e}-b_{s})^{2}\left(a_{\mathsf{c}(e)}^{4}-2a_{\mathsf{c}(e)}^{2}v_{s}^{2}-3v_{s}^{4}\right)}{\left(v_{s}^{2}+a_{\mathsf{c}(e)}^{2}\right)^{4}}\label{eq:vs2nd}\\
\frac{\partial^{2}L(\theta)}{\partial a_{c}^{2}} & = & \sum_{s,e:\mathsf{c}(e)=c}-\frac{v_{s}^{2}-a_{\mathsf{c}(e)}^{2}}{\left(v_{s}^{2}+a_{\mathsf{c}(e)}^{2}\right)^{2}}+\frac{(x_{e,s}-x_{e}-b_{s})^{2}\left(v_{s}^{4}-2v_{s}^{2}a_{\mathsf{c}(e)}^{2}-3a_{\mathsf{c}(e)}^{4}\right)}{\left(v_{s}^{2}+a_{\mathsf{c}(e)}^{2}\right)^{4}}\label{eq:ac2nd}
\end{eqnarray}
Define $w_{e,s}=\frac{1}{v_{s}^{2}+a_{\mathsf{c}(e)}^{2}}$. Applying
the expressions above to the Newton-Raphson update rules yield:

\begin{eqnarray*}
x_{e}^{new} & = & \frac{\sum_{s}w_{e,s}(x_{e,s}-b_{s})}{\sum_{s}w_{e,s}}\\
b_{s}^{new} & = & \frac{\sum_{e}w_{e,s}(x_{e,s}-x_{e})}{\sum_{e}w_{e,s}}\\
v_{s}^{new} & \text{=} & v_{s}-\frac{\sum_{e}w_{e,s}v_{s}-w_{e,s}^{2}v_{s}(x_{e,s}-x_{e}-b_{s})^{2}}{\sum_{e}w_{e,s}^{2}(a_{\mathsf{c}(e)}^{2}-v_{s}^{2})-w_{e,s}^{4}(x_{e,s}-x_{e}-b_{s})^{2}\left(a_{\mathsf{c}(e)}^{4}-2a_{\mathsf{c}(e)}^{2}v_{s}^{2}-3v_{s}^{4}\right)}\\
a_{c}^{new} & = & a_{c}-\frac{\sum_{s,e:\mathsf{c}(e)=c}w_{e,s}a_{\mathsf{c}(e)}-w_{e,s}^{2}a_{\mathsf{c}(e)}(x_{e,s}-x_{e}-b_{s})^{2}}{\sum_{s,e:\mathsf{c}(e)=c}w_{e,s}^{2}(v_{s}^{2}-a_{\mathsf{c}(e)}^{2})-w_{e,s}^{4}(x_{e,s}-x_{e}-b_{s})^{2}\left(v_{s}^{4}-2v_{s}^{2}a_{\mathsf{c}(e)}^{2}-3a_{\mathsf{c}(e)}^{4}\right)}
\end{eqnarray*}
Note that there is strong intuition behind the expressions
for $x_{e}^{new}$ and $b_{s}^{new}$. In each iteration, $x_{e}^{new}$
is re-estimated, by the sum of opinion scores $x_{e,s}$ with the
currently estimated bias $b_{s}$ removed. Each opinion score is weighted
by $w_{e,s}$, i.e., the higher the variance (subject inconsistency
and content ambiguity), the less reliable the opinion score, hence
less the weight. We can interpret $b_{s}^{new}$ in a similar way.

\section{Derivation of Confidence Interval (\ref{eq:ci})\label{sec:confidence-interval}}

Let $I(\theta)$ be the Fisher information of parameter $\theta$ defined by 
\begin{eqnarray*}
I(\theta)=E\left[\left(\frac{\partial L(\theta)}{\partial\theta}\right)^{2}\right]=-E\left[\frac{\partial^{2}L(\theta)}{\partial\theta^{2}}\right]
\end{eqnarray*}
where the second equation is true if $L(\theta)$ is twice-differentiable and $E\left[\frac{\partial L(\theta)}{\partial\theta}\right]=0$. The Cramer-Rao bound states that the variance of $\hat{\theta}$ is lower-bounded by the reciprocal of $I(\theta)$:
\begin{eqnarray*}
\var(\hat{\theta})\geq\frac{1}{I(\theta)}, 
\end{eqnarray*}
where the equality holds for the Gaussian case. Using the Gaussian assumption and observed Fisher information, we have $\var(\hat{\theta})=\frac{1}{-\frac{\partial^{2}L(\theta)}{\partial\theta^{2}}}$.
Substituting into the 95\% confidence interval $\hat{\theta}\pm1.96\std(\hat{\theta})$, we have (\ref{eq:ci}).

\end{document}